\documentclass{ifacconf}

\let\theoremstyle\relax
\usepackage{graphicx}      
\usepackage{natbib}        
\usepackage{amsmath}
\usepackage{amsfonts}
\usepackage{amsthm}
\usepackage[T1]{fontenc}
\newtheorem{theorem}{Theorem}
\newtheorem{lemma}{Lemma}

\begin{document}
\begin{frontmatter}

\title{Fast Verification of Control Barrier Functions via Linear Programming } 

\thanks[footnoteinfo]{
This work was supported by AFOSR under grant FA9550-19-1-0169,
ONR under grants
N00014-19-1-2543,
N00014-21-2-2502,
and
N00014-22-1-2435,
and AFRL under grants
FA8651-22-F-1052
and
FA8651-23-FA-006. 
}

\author[First]{Ellie Pond} 
\author[Second]{Matthew Hale} 

\address[First]{University of Florida, Gainesville, FL 32611 USA (email: elliempond@ufl.edu).}

\address[Second]{University of Florida, Gainesville, FL 32611 USA (email: matthewhale@ufl.edu).} 

\begin{abstract}                
Control barrier functions are a popular method of ensuring system safety, and 
these functions can be used to 
enforce invariance of a set under the dynamics of a system. 
A control barrier function must have certain properties, and one
must both formulate a candidate control barrier function and verify that
it does indeed satisfy the required properties. Targeting the latter problem,
this paper presents a method of verifying any finite number of candidate control barrier functions with linear programming.  We first apply techniques from real algebraic geometry to formulate verification problem statements that are solvable
numerically. 
Typically, semidefinite programming is used to verify candidate control barrier functions, but this does not always scale well. 
Therefore, we next apply a method of inner-approximating the set of sums of squares polynomials that significantly reduces the computational complexity of these verification problems by transcribing them to linear programs.
We give explicit forms for the resulting linear programs, and
simulation results for a satellite inspection problem 
show that the computation time needed for verification
can be reduced by more than~$95$\%. 
\end{abstract}


\end{frontmatter}

\section{Introduction}
Interest in safety-critical autonomous systems has grown in recent years. 
For example, mobile autonomous robots must avoid collisions with each other \cite{Zohaib_RobotObstacleAvoidanceSurvey,Kunchev_ObstacleAvoidance}
and co-robots working alongside humans must not collide with them \cite{Zacharaki_SafetyHumanRobotSurvey,Vasic_HumanRobotInteraction}. Model predictive control \cite{Gao_MPCObstacleAvoidance}, reachability analysis \cite{Althoff_Reachability}, and ``back-up" control strategies \cite{Roumeliotis_SecondaryControl} are among the methods that address safety. However, these approaches may suffer from poor performance or from being overly conservative.
One alternative way to formalize system safety is by defining a safe region of a system's
state space and ensuring that the system's dynamics render that set forward-invariant in time. 

A common method of ensuring this type of safety through invariance is 
with control barrier functions (CBFs).
A representative sample of works using
CBFs for safety is \cite{Xiao_HOCBFS,Jankovic_CollisionAvoidance,Landi_HumanRobotInteraction,Xu_RobustnessOfCBFs,Cortez_RobustMultiCBFFramework}. 
A CBF is a function that is positive on the safe set and whose derivative
satisfies a certain bound on the safe set. 
A CBF is then used in 
controller design to impose 
inequality constraints on the system input that
ensure invariance of the safe set. 
These inequality constraints are in terms of the control barrier function itself, which
must be found and verified to satisfy the conditions that make it a valid CBF before it can be
used in the control design process. 

The need for verification has led to the development
of several computational approaches that are analogous
to the verification of Lyapunov
functions. Recent work in this vein has used
semidefinite programming and 
sums of squares (SOS) optimization
to validate that a candidate CBF is indeed a valid
CBF \cite{Wang_PermissiveBarrierCertificates,Isaly,Clark_Verification}. 
However, these verification approaches may not scale well
to systems with large state spaces, e.g., bipedal robot control \cite{Feng_AtlasRobot},
or systems with many constraints CBFs \cite{Wang_SafetyBarrierCertificates}.

To address this need, in this paper we develop scalable verification techniques for systems with single and multiple CBFs.
Our approach leverages tools from real algebraic
geometry to 
formulate computational problems to assess 
the nonnegativity properties of a candidate CBF. 
We then apply the method introduced in \cite{Ahmadi_DSOSAndSDSOSOptimization} to translate these verification statements to linear programs (LPs).

To summarize, our contributions are:
\begin{itemize}
\item Derivation of CBF verification statements for dynamical systems with any finite number of candidate CBFs
\item Explicit mapping of the verification statements to linear programs
\item Simulations that demonstrate the substantial improvement in scalability and
reduction of computation time in practice.
\end{itemize}

There are several works that propose verification methods for systems with multiple candidate CBFs. In \cite{Cortez_RobustMultiCBFFramework}, a neighborhood notion is employed to switch priority of  CBFs to circumvent possible conflicts. In \cite{Black_AdaptationForValidationOfConsolidatedCBF}, a method of consolidating multiple CBFs to one candidate CBF through under-approximating the intersection of safe sets is proposed. The most relevant work can be found in \cite{Wang_PermissiveBarrierCertificates,Isaly,Clark_Verification}, where algebraic geometry tools are applied to create SOS validation programs. However, we apply different results from algebraic geometry and develop linear program translations that allow for improved scaling and decreased computation times relative to SOS-based works. To the best of our knowledge, CBF verification has not been addressed through linear programming before.

The rest of this paper is organized as follows. In Section 2, we provide necessary background on CBFs, Positivstellens\"atze from algebraic geometry, and DSOS programs, as well as the problem statements solved in this work. In Section 3, we develop verification statements for single and multiple candidate CBFs through application of algebraic geometry. In Sections 4 and 5, we explicitly state the translation from verification statement to linear programs for single and multiple CBF problems, respectively. In Section 6, we present simulation results that demonstrate the benefit of linear programming over semidefinite programming for verification. In Section 7, we conclude.

\section{Background and Problem Formulation}
In this section, we introduce background on CBFs, Positivstellens\"atze, and DSOS programs. We then provide three formal problem statements that will be the focus of the remaining sections.

Throughout, for indexing up to an integer $N$ we use $i\in[N]=\{1,\dots,N\}$. For functions $f$ and $g$, we define the Lie derivative of $f$ with respect to $g$ as $L_gf(x)=\frac{\partial f(x)}{\partial x}g(x)$. We use $\mathbb{R}[x]$ to denote the set of scalar-valued polynomials with real coefficients. Here, $x=(x_1,\dots,x_n)^T$ is a vector of indeterminates whose dimension will be clear from context. We define a polynomial matrix $p\in\mathbb{R}^{m\times n}[x]$ as an $m \times n$ matrix with polynomial entries $p_{ij}\in\mathbb{R}[x]$ for all $i\in[m]$, $j\in[n]$. We write $\mathbb{R}^{m}[x]$ for $\mathbb{R}^{m\times 1}[x]$. We define an operation $\text{col}:\mathbb{R}^{m\times n}\rightarrow \mathbb{R}^{mn}$ as the mapping from a matrix to a column vector containing the transposed rows of the matrix. For example, for a matrix 
\begin{equation}
    p=\begin{bmatrix} p_{11} & p_{12} &\hdots & p_{1n} \\ p_{21} & p_{22} & \hdots & p_{2n} \\ \vdots & & \ddots& \vdots \\ p_{m1} & p_{m2} & \hdots &p_{mn} \end{bmatrix}
\end{equation}
with $p_{ij}\in\mathbb{R}[x]$ for all $i\in[m]$, $j\in [n]$, we have 
\begin{equation}
         \text{col}(p)= \begin{bmatrix}p_{11} & p_{12} & \dots & p_{1n} & \dots & p_{m1} & p_{m2} & \dots & p_{mn}\end{bmatrix}^T.   \\
\end{equation}

\subsection{Control Barrier Functions}

We will consider dynamical systems that are affine in the control and take the form
\begin{equation}\label{sysdef}
    \dot{x}=f(x)+g(x)u,
\end{equation}
where $x\in\mathbb{R}^n$ is the state of the system, $u\in  \mathbb{R}^m$ is the input, and $f\in\mathbb{R}^n[x]$, $g\in\mathbb{R}^{n\times m}[x]$ define the system dynamics.

The polynomial $b\in \mathbb{R}[x]$ defines a safe set $\mathcal{B}=\{x\in\mathbb{R}^n \mid b(x)\geq 0\}$.  The function $b$ is a CBF if there exists an extended class-$\mathcal{K}$ function, $\kappa$, such that \cite{Ames_TheoryApplications}
\begin{equation}
\begin{aligned}
    \sup_{u} \Big\{ L_fb(x) + L_gb(x)u + \kappa(b(x)) \Big\}&\geq 0, \\ &\text{for all } x\in \mathbb{R}^n.
\end{aligned}
\end{equation}
The Lie derivatives of $b$ are $L_fb\in \mathbb{R}[x]$ and $L_gb\in \mathbb{R}^m[x]$. CBFs can be used to impose an inequality constraint on the input that takes the form
\begin{equation}\label{CBFineq}
    L_gb(x)u\geq \kappa(b(x)) - L_fb(x).
\end{equation}

Control barrier functions are used to ensure safety for a system. In this capacity, safety is equivalent to forward invariance of the system's safe set $\mathcal{B}$. The set $\mathcal{B}$ is forward invariant if for every $x_0\in\mathcal{B}$, we have $x(t)\in\mathcal{B}$ for all $t\geq 0$. If a control input exists that satisfies (\ref{CBFineq}), then the safe set $\mathcal{B}$ can be rendered forward invariant and the system is rendered safe  \cite{Ames_TheoryApplications}.

In this work we focus on systems with unbounded control inputs. A candidate CBF $b$ satisfies (\ref{CBFineq}) if and only if there are no points $x\in\mathbb{R}^n$ such that $L_fb(x)<0$ while $b(x)=0$ and $L_gb(x)=0$ \cite{Clark_Verification}. When $b(x)=0$, a control input can be chosen so that (\ref{CBFineq}) is satisfied. However, if both $b(x)=0$ and $L_gb(x)=0$, then the dynamics and candidate CBF must have $L_fb(x)\geq0$.

\subsection{Positivstellens\"{a}tze}
We consider a finite set of polynomials
\begin{equation}
    S = \{f_1,\dots,f_r;g_1,\dots,g_s\},
\end{equation}
where $f_i,g_j\in \mathbb{R}[x]$ for all $i\in[r]$ and $j\in[s]$. The basic semialgebraic set $K$ associated to $S$ is 
\begin{equation}\label{Kdef}
\begin{aligned}
    K= \{x\in \mathbb{R}^n \; | \; f_i(x)&\geq0, \;  i = 1,\dots,r; \;\; \\ &g_j(x)=0, \; j=1,\dots,s\}.
\end{aligned}
\end{equation}
We denote by $\Sigma[x]$ the set of all sums of squares (SOS) polynomials in $\mathbb{R}[x]$. All $s\in\Sigma[x]$ take the form $s(x) = \sum_i p_i^2(x)$ for some $p_i\in\mathbb{R}[x]$. 

Preorders and quadratic modules are algebraic objects associated to the set of polynomials $S$. Let the subset $P_S\subseteq\mathbb{R}[x]$ be the smallest preorder containing the set $S$. Then a polynomial $q\in P_S$ if and only if $q$ is of the form
\begin{equation}\label{PSdef}
    q = \sum_{\epsilon\subseteq\{0,1\}^r} s_\epsilon (x) f_\epsilon(x) + \sum_{j=1}^s p_j(x)g_j(x),
\end{equation}
where $p_j\in\mathbb{R}[x]$ for $j\in[s]$ and each $s_\epsilon \in\Sigma[x]$, with $f_\epsilon = f_1^{\epsilon_1}\cdot\dots \cdot f_r^{\epsilon_r}$ and $f_{\emptyset} = 1$ \cite{Powers_CertificatesOfPositivity}. 

Let the subset $M_S\subseteq\mathbb{R}[x]$ be the smallest quadratic module containing the set $S$. Then a polynomial $q\in M_S$ if and only if $q$ is of the form
\begin{equation}\label{MSdef}
    q = s_0(x) + \sum_{i=1}^r s_i(x)f_i(x) + \sum_{j=1}^s p_j(x)g_j(x),
\end{equation}
where $p_j\in \mathbb{R}[x]$ for $j\in[s]$ and $s_i\in\Sigma[x]$ for $i\in[r]\cup\{0\}$ \cite{Powers_CertificatesOfPositivity}. 

\textit{Remark 1 (\cite{Powers_AlgorithmForSOS})}. The quadratic module of $S$ is always contained in the preorder, i.e., $M_S\subseteq P_S$. When $r\leq 1$ in $K$, we have $M_S=P_S$. 

\textit{Proposition 1 (Archimedean Property \cite{Powers_CertificatesOfPositivity})}. Fix a polynomial ring $\mathbb{R}[x]$ over the indeterminates $x=(x_1,\dots,x_n)^T$, and fix a set of polynomials $S = \{f_1,\dots,f_r;g_1,\dots,g_s\}$. Then, for a quadratic module $M_S$, the following conditions are equivalent:
\begin{enumerate}
    \item $M_S$ is Archimedean. 
    \item There exists $C\in\mathbb{N}\backslash \{0\}$ such that $C- \sum_{i=1}^nx_i^2\in M_S$.
    \item There exists $C\in\mathbb{N}\backslash \{0\}$ such that both $C\pm x_i\in M_S$ for $i=1,\dots,n$.
\end{enumerate}
The Archimedean property is slightly stronger than $K$ being compact \cite{Ahmadi_DSOSAndSDSOSOptimization}. 

The Positivstellens\"{a}tze are a class of theorems in real algebraic geometry that characterize positivity properties of polynomials over semialgebraic sets. The broadest variation is credited to Krivine and Stengle.

\textit{Proposition 2 (KS Positivstellensatz \cite{Bochnak_RealAlgebraicGeometry})}. Fix a set $S = \{f_1,\dots,f_r;g_1,\dots,g_s\}$. Let $K$ and $P_S$ be as defined in (\ref{Kdef}) and (\ref{PSdef}), and let $q\in\mathbb{R}[x]$. Then
\begin{enumerate}
    \item $q>0$ on $K$ if and only if there exist $h_1,h_2\in P_S$ such that $h_1q=1+h_2$.
    \item $q\geq 0$ on $K$ if and only if there exist an integer $m\geq 0$ and $h_1,h_2\in P_S$ such that $h_1q=q^{2m}+h_2$.
    \item $q=0$ on $K$ if and only if there exists an integer $m\geq 0$ such that $-q^{2m}\in P_S$.
    \item $K=\emptyset$ if and only if $-1\in P_S$.
\end{enumerate}
These four statements are equivalent. 

Another variation of the Positivstellensatz is credited to Putinar. This version makes use of the quadratic module through introducing an assumption. Using the quadratic module in computations is advantageous to using the preorder, due to the linear growth rate of the number of SOS terms in expressions in the quadratic module compared to the exponential growth rate of the number of SOS terms when using expressions in the preorder. 

\textit{Proposition 3 (Putinar's Positivstellensatz \cite{Powers_CertificatesOfPositivity})}. Fix a set $S = \{f_1,\dots,f_r;g_1,\dots,g_s\}$. Let $K$ and $P_S$ be as defined in (\ref{Kdef}) and (\ref{PSdef}), and let $q\in\mathbb{R}[x]$. Assume that $M_S$ is Archimedean, satisfying Proposition 1. Then $q>0$ on $K$ implies that $q\in M_S$.

Lastly, a corollary credited to Jacobi is introduced. We note that a quadratic module in $\mathbb{R}[x]$ can alternatively be defined as a $\Sigma[x]$-module, where $\Sigma[x]$ is a generating preprime in $\mathbb{R}[x]$, i.e., $\Sigma[x]-\Sigma[x]=\mathbb{R}[x]$. We present the following proposition in terms of the polynomial ring $\mathbb{R}[x]$ because that is what we will consider going forward, though it generally applies to any commutative ring.

\textit{Proposition 4 (Jacobi's Representation Corollary \cite{Jacobi_RepresentationTheorem})}. Suppose $T$ is a generating preprime in a commutative ring $\mathbb{R}[x]$, $M_S$ is an Archimedean $T$-module, and $K$ is the nonnegativity set defined by (\ref{Kdef}). Then $-1\in M_S$ if and only if $K = \emptyset$.

\subsection{DSOS Programming}
Traditionally, semidefinite programs (SDPs) have been used to solve problems originating from Positivstellens\"{a}tze or polynomial optimization problems with nonnegativity constraints. Through this method, nonnegativity is approximated by a more restrictive but computationally tractable SOS requirement. 

We denote by $\mathcal{P}[x]$ the set of all nonnegative polynomials in $\mathbb{R}[x]$. For indeterminates $x\in\mathbb{R}^n$, let $m(x)\in\mathbb{R}^k$ be a monomial basis vector of degree $d$, where $k=\binom{n+d}{n}$. Let $s\in\mathbb{R}[x]$ be of degree $2d$. Then $s$ is an SOS polynomial if and only if there exists a symmetric Gram matrix, $Q$, such that \cite{Powers_AlgorithmForSOS}
\begin{equation}
    s(x) = m^T(x)Qm(x) \quad \text{and}
\end{equation}
\begin{equation}
    Q\succeq 0.
\end{equation}
Searching for a $Q$ that is positive semidefinite (PSD) is conventionally done through semidefinite programming. 

A matrix $Q\in\mathbb{R}^{k\times k}$ is diagonally dominant if 
\begin{equation}\label{DD}
    Q_{ii}\geq \sum_{j\in[k]\backslash\{i\}} | Q_{ij}| \quad\quad \text{for all } i\in[k].
\end{equation}
We define `$\succeq_{dd}$' to indicate diagonal dominance of a matrix. That is, $Q\succeq_{dd}0$ denotes that $Q$ satisfies (\ref{DD}) and thus
is positive semidefinite. 
A polynomial $s\in\mathbb{R}[x]$ is a diagonally dominant sums-of-squares (DSOS) polynomial if it can be written as 
\begin{equation}
\begin{aligned}
    s(x) = \sum_{i} a_im_i^2(x) + \sum_{i,j}&b_{ij}(m_i(x) + m_j(x))^2  \\ + &\sum_{i,j}c_{ij}(m_i(x)-m_j(x))^2
\end{aligned}
\end{equation}
for some monomials $m_i(x),m_j(x)$ and some nonnegative scalars $a_i,b_{ij},c_{ij}$.

\textit{Proposition 5 (\cite{Ahmadi_DSOSAndSDSOSOptimization})}. Let $s\in\mathbb{R}[x]$ be of degree $2d$. Then $s$ is a diagonally dominant sum of squares (DSOS) polynomial if and only if there exists a symmetric Gram matrix, $Q$, such that 
\begin{equation}
    s(x) = m^T(x)Qm(x) \quad \text{and}
\end{equation}
\begin{equation}
    Q\succeq_{dd}0.
\end{equation}

Computationally, the condition $Q\succeq_{dd}0$ is enforced using the linear constraints
\begin{equation}\label{DDconst}
    \begin{aligned}
    -Q_{ii} + \sum_{j\in[k]\backslash\{i\}}\tau_{ij} &\leq  0 \quad \quad \text{for all } i \in [k]\\
    Q_{ij} - \tau_{ij} &\leq 0 \quad \quad \text{for all } i,j\in [k],\; i\not=j\\
    -Q_{ij} - \tau_{ij} &\leq 0 \quad \quad \text{for all } i,j\in [k],\; i\not=j,\\
    \end{aligned}
\end{equation}
where $\tau\in\mathbb{R}^{k\times k}$ is a symmetric bounding matrix that enforces the absolute value constraints. We note that $\tau_{ij}\geq0$ for $i,j\in[k]$ is implicitly enforced by (\ref{DDconst}). We denote by $\mathcal{D}[x]$ the set of all DSOS polynomials in $\mathbb{R}[x]$. The inclusion relationship between the sets of polynomials is $\mathcal{D}[x]\subseteq \Sigma[x]\subseteq \mathcal{P}[x]\subseteq \mathbb{R}[x]$.

\textit{Proposition 6 (\cite{Ahmadi_DSOSAndSDSOSOptimization})}. For a fixed degree $d$, optimization over $\mathcal{D}[x]$ can be done through a linear program (LP).

In such an LP, the polynomial variables take the form $p(x)=c^Tm(x)$. 
We define the operation $\text{uTri}:\mathbb{R}^{k\times k}\rightarrow \mathbb{R}^{\frac{1}{2}(k^2+k)}$ as the mapping from a symmetric matrix to a vector of the upper trianglular elements of the matrix. For example, for $k=3$ the matrix
\begin{equation}
    Q = \begin{bmatrix} q_{11} & q_{12} & q_{13} \\ q_{12} & q_{22} & q_{23} \\ q_{13} & q_{23} & q_{33} \end{bmatrix}
\end{equation}
gives
\begin{equation}
    \text{uTri}(Q) = \begin{bmatrix} q_{11} & q_{12} & q_{13} & q_{22} & q_{23} & q_{33}  \end{bmatrix}^T.
\end{equation}
In DSOS programs with $I$ DSOS polynomial variables of the form $m^T(x)Q_im(x)$ for $i\in[I]$ and $J$ polynomial variables of the form $c_j^Tm(x)$ for $j\in [J]$, the decision vector is
\begin{equation}
\begin{aligned}
    z = [c_1^T \;\; \hdots \;\; c_J^T \;\; \text{uTri}&(Q_1)^T \;\; \text{uTri}(\tau_1)^T \;\; \hdots \;\; \\ &\text{uTri}(Q_I)^T  \;\;  \text{uTri}(\tau_I)^T]^T
\end{aligned}
\end{equation}
with size $z\in\mathbb{R}^{I(k^2+k)+Jk}$.

Lastly, we define the operation $\text{G}_k:\mathbb{R}[x]\rightarrow\mathbb{R}^{k\times k}$ as the mapping from a polynomial expression to the corresponding Gram matrix, where $k$ is the degree of the monomial basis. For any polynomial $e(x)$, there exists a Gram matrix representation of the form $e(x)=m^T(x)G_k(e)m(x)$, where $G_k(e)$ is composed of the monomial coefficients of $e$. For example, consider polynomials $p(x) = c_1 + c_2x + c_3 x^2$ and $b(x) = x^2-4$, with indeterminate $x\in\mathbb{R}$ and monomial degree $k=4$. Then
\begin{equation}\label{gram}
    \text{G}_4(pb) = \begin{bmatrix}-4c_1 & -4c_2 & c_1-4c_3 & c_2 & c_3\\ 0 & 0 & 0 & 0 & 0\\ 0 & 0 & 0 &0 & 0 \\ 0 & 0 & 0 &0 &0  \end{bmatrix}.
\end{equation}
The output of the Gram operation satisfies $p(x)b(x) = m^T(x)G_4(pb)m(x)$ with $m(x) = [1 \;\; x\;\; x^2 \;\; x^3 \;\; x^4]^T$. This operation removes dependence on the indeterminate $x\in\mathbb{R}$ and enables the development of conditions that apply for all $x\in\mathbb{R}^n$, e.g., positive definiteness. This property will be utilized in the linear programs we present in Sections 4 and 5.

\subsection{Problem Formulation}

In this section, we state the problems we solve in the remainder of the paper.

\subsubsection{Problem 1.} Given a collection of candidate CBFs, $S=\{b_1,\dots,b_L\}$ with $b_i\in\mathbb{R}[x]$ for $i\in[L]$, determine necessary and sufficient conditions for safety using the Positivstellens\"{a}tze.

For this problem, we will clearly establish the connection between real algebraic geometry and candidate CBF verification. We will make explicit the motivation behind the theorems we apply and the necessity of the assumptions placed. In Section 3, we begin by considering the verification of a single candidate CBF. This is then extended to the verification of multiple candidate CBFs.

\subsubsection{Problem 2.} Given a system (\ref{sysdef}) and a candidate CBF, $b\in\mathbb{R}[x]$, verify system safety through forward invariance of the safe set $\mathcal{B}=\{x\in\mathbb{R}^n\mid b(x)\geq 0\}$ with a linear program. 

\subsubsection{Problem 3.} Given a system (\ref{sysdef}) and a collection of candidate CBFs, $S=\{b_1,\dots,b_L\}$ with $b_i\in\mathbb{R}[x]$ for $i\in[L]$, verify system safety through forward invariance to the system's safe set $\mathcal{B}=\{x\in\mathbb{R}^n\mid b_1(x)\geq 0,\dots,b_L(x)\geq 0\}$ with linear programming.

In Sections 4 and 5, we show the translation from the verification statements derived in Section 3 to linear programs. While these verification statements are nonlinear in the indeterminate variable $x\in\mathbb{R}^n$, they form optimization programs that are linear over their decision variables.

\section{Safety Verification with Positivstellens\"atze}

In this section we solve Problem 1. We begin by considering systems with a single candidate CBF. Several notions from algebraic geometry can be applied to determine nonnegativity properties of a candidate CBF, and while these variations can all theoretically verify candidate functions, their numerical implementation can vary significantly in the size of the corresponding problems and in the ease of computation of their solutions. In determining the nonegativity theorems we apply to create verification statements, we keep in mind the end goal of scaling to large numbers of CBFs and translating to implementable optimization problems. This leads us to the following.

\begin{theorem}
Given a system in (\ref{sysdef}) and a single candidate CBF, $b\in\mathbb{R}[x]$, let $S=\{b,L_gb\}$ and $K=\{x\in\mathbb{R}^n \mid b(x)=0,L_gb(x)=0\}$. Then $b$ is a CBF and the system is verified forward invariant to $\mathcal{B}$ if and only if there exists $a\in\mathbb{N}_{\geq0}$, SOS polynomials $s_1,s_2\in\Sigma[x]$, and polynomials $p_{10},p_{20}\in\mathbb{R}[x]$ and $p_1,p_2\in\mathbb{R}^m[x]$ such that
\begin{equation}\label{Theorem1}
    h_1(x)L_fb(x) - h_2(x) - (L_fb(x))^{2a}=0
\end{equation}
and such that $h_1,h_2\in M_S$ are of the form
\begin{equation}\label{pivec}
\begin{aligned}
    h_i(x) = s_i(x) + b(x)p_{i0}(x) + L_gb(x)&p_{i}(x), \quad  
    \\& \text{for } i=1,2,
\end{aligned}
\end{equation}
as defined in (\ref{MSdef}).
\end{theorem}

\begin{proof}
In order to verify safety, the candidate CBF $b$ must satisfy (\ref{CBFineq}). As introduced in Section 2.1, (\ref{CBFineq}) is satisfied over the state $x\in\mathbb{R}^n$ if and only if $L_fb(x)\geq0$ over the set $K$.

Applying the KS Positivstellensatz in Proposition 2, we have $L_fb(x)\geq0$ over $K$ if and only if there exists $a\in\mathbb{N}_{\geq0}$ and $h_1,h_2\in P_S$ such that (\ref{Theorem1}) is satisfied. Here, we have defined the set $S$ with $r=0$ inequalities and $s=2$ equations. As noted in Remark 1, since $r\leq1$ we have $P_S=M_S$. Since $h_1,h_2\in M_S$ they must take the form of (\ref{MSdef}), which leads to (\ref{pivec}). Then $L_fb(x)\geq0$ over $K$ if and only if
\begin{equation}\label{Proof1a}
    h_1(x)L_fb(x) - h_2(x) -(L_fb(x))^{2a} = 0
\end{equation}
for the $h_1,h_2\in M_S$ and $a\in\mathbb{N}_{\geq0}$ above.
\end{proof}

This theorem uses the same motivation as that in \cite{Clark_Verification} to prove forward invariance. However, we  apply an alternate variation of the Positivstellensatz to form a verification statement that is well-suited to numerical implementation via linear programming and for scaling the number of candidate CBFs. This difference arises in considering the nonnegativity of $L_fb$ over the set $K$, as opposed to testing set emptiness of an augmented $S$. The benefit of this alternate approach materializes in the next theorem. 

A natural extension is safety verification for systems with $L$ candidate CBFs. Each CBF $b_i$ corresponds to a basic semialgebraic set $K_i = \{x\in\mathbb{R}^n \mid b_i(x)=0,L_gb_i(x)=0\}$ and an independent safe set $\mathcal{B}_i=\{x\in\mathbb{R}^n\mid b_i(x)\geq 0\}$. The system safe set $\mathcal{B}$ is the intersection of all independent safe sets, namely $\mathcal{B}=\cap_{i=1}^L\mathcal{B}_i$.

\begin{theorem}
Given a system defined by (\ref{sysdef}) and a collection of candidate CBFs,  $S=\{b_1,\dots,b_L\}$, where $b_i\in\mathbb{R}[x]$ for $i=1,\dots,L$, assume that $M_S$ satisfies the Archimedean property. Then $b_i\in S$ is a CBF for all $i\in[L]$ and  $\mathcal{B}\not=\emptyset$ if and only if the following conditions are satisfied:
\begin{enumerate}
    \item Each of the $L$ candidate CBFs in $S$ satisfies Theorem 1.
    \item There do not exist SOS polynomials $s_i\in\Sigma[x]$ for $i=0,\dots,L$ such that 
    \begin{equation} \label{MultiC2}
        1+ s_0(x) + \sum_{i=1}^Ls_i(x)b_i(x) = 0.
    \end{equation}
\end{enumerate}
\end{theorem}

\begin{proof} 
If each of the $L$ candidate CBFs satisfies the conditions of Theorem 1, then all $b_i\in S$ are CBFs individually. This is shown by applying Theorem 1 to each $b_i$. Thus, if Condition $1$ in the theorem statement is satisfied, then for all $i\in[L]$ the CBF $b_i$ implies the existence of an input that renders the system forward invariant to the safe set $\mathcal{B}_i = \{x\in\mathbb{R}^n \mid b_i(x)\geq 0 \}$.

When there are $L>1$ CBFs in a system, there must exist a solution to $L$ copies of (\ref{CBFineq}) simultaneously. That is, we require the existence of an input $u\in\mathbb{R}^m$ that satisfies the $L$ inequalities
\begin{equation}
    \begin{aligned}
    L_gb_1(x)u&\geq \kappa_1(b_1(x)) - L_fb_1(x) \\
    \vdots \\
    L_gb_L(x)u&\geq \kappa_L(b_L(x)) - L_fb_L(x). \\
    \end{aligned}
\end{equation}
Condition 2 in the theorem statement then determines if such a~$u$
exists and thus determines 
if the intersection set $\mathcal{B} = \cap_{i=1}^L\mathcal{B}_i$ is empty. Specifically, if Condition 2 is satisfied, then there is no solution to (\ref{MultiC2}) and the constant polynomial $q=-1$ is not of the form of (\ref{MSdef}). Then, $-1\not\in M_S$. As introduced in Section 2.2, the quadratic module $M_S$ is a $\Sigma[x]$-module and $\Sigma[x]$ is a generating preprime in $\mathbb{R}[x]$. Applying Proposition 4, since $M_S$ is Archimedean, if $-1\not\in M_S$ then $\mathcal{B}\not=\emptyset$. Thus, satisfaction of Conditions 1 and 2 implies $\mathcal{B}\not=\emptyset$.

To prove the other direction, now suppose that $\mathcal{B}\not=\emptyset$. Clearly $\mathcal{B}\not=\emptyset$ implies $\mathcal{B}_i\not=\emptyset$ for all~$i \in [L]$, which implies that Condition 1 holds for all~$i \in [L]$. Next, by applying the KS Positivstellensatz defined in Proposition 2, $-1\in P_S$ if and only if $\mathcal{B}=\emptyset$. Thus, $\mathcal{B}\not=\emptyset$ implies $-1\not\in P_S$. From Remark 1, we have $M_S\subseteq P_S$ and thus $-1\not\in P_S$ implies $-1\not\in M_S$. For $S=\{b_1,\dots,b_L\}$, a polynomial $q\in M_S$ if and only if $q$ satisfies (\ref{MSdef}). The constant polynomial $q=-1\not\in M_S$, therefore there cannot exist SOS polynomials $s_0,s_1,\dots,s_L$ such that (\ref{MultiC2}) is satisfied. Thus $\mathcal{B}\not=\emptyset$ implies that Conditions 1 and 2 hold, and the proof is complete.
\end{proof}

\section{Single CBF Linear Program}
In this section we solve Problem 2. We address the translation of the verification statement derived in Theorem 1 for a single candidate CBF to a linear program. 
The SDPs/LPs created from translating the Positivstellens\"atze statements can involve both inequality and equality constraints containing the SOS/DSOS polynomial decision variables. Inequality constraints of the form $e(x)\leq0$ are taken to be $-e(x)\in\mathcal{D}[x]$. This creates constraints of the form $-e(x)=m^T(x)G_k(e)m(x)$ and $G_k(e)\succeq_{dd}0$. Let $d$ be the degree of $e$, then $k=\lceil d/2\rceil$.

There are a few methods for handling the equality constraints of the form $e(x)=0$. One method is to constrain both $e(x)\in\mathcal{D}[x]$ and $-e(x)\in\mathcal{D}[x]$ \cite{Ahmadi_DSOSAndSDSOSOptimization}. However, this can cause numerical difficulties. Alternatively, $e(x)=m^T(x)G_k(e)m(x)$ with $G_k(e)=0$ can be used to impose equality. Here, $G_k(e)=0$ is interpreted element-wise. The following theorem uses this fact and ideas from \cite{Ahmadi_DSOSAndSDSOSOptimization} to create a linear program.

\begin{theorem}
Given a system (\ref{sysdef}) and a single candidate CBF $b\in\mathbb{R}[x]$, let $a\geq0$ be an integer, $p_{10},p_{20}\in\mathbb{R}[x]$ and $p_1,p_2\in\mathbb{R}^m[x]$ be polynomials, and $s_1,s_2\in\mathcal{D}[x]$ be DSOS polynomials. Let $G_k$ be as defined by (\ref{gram}) with $k=\lceil d/2 \rceil$, where $d$ is the degree of its polynomial argument. If the following LP has a solution, then $b$ is a CBF:
\begin{equation}\label{Theorem3}
\begin{aligned}
    \min_{z} \quad &0 \\
    \text{s.t.} \quad  
    &Q_i \succeq_{dd}0 \quad \quad \text{for } i=1,2\\
    &G_k((s_1 + bp_{10} + L_gbp_1)L_fb - s_2 \\ &\quad \quad - bp_{20} - L_gbp_{2} - (L_fb)^{2a})=0,
\end{aligned}
\end{equation}
where the decision variable is
\begin{equation}
\begin{aligned}
    z = [c_{10}^T \;\;  c_{20}^T \;\; &\text{col}(c_{1})^T \;\; \text{col}(c_{2})^T \;\; \text{uTri}(Q_1)^T \;\; \\ &\text{uTri}(\tau_1)^T \;\; \text{uTri}(Q_2)^T \;\; \text{uTri}(\tau_2)^T ]^T,
\end{aligned}
\end{equation}
with $z\in\mathbb{R}^{2k^2+(2m+4)k}$, the DSOS variables take the form $s_i(x) = m^T(x)Q_im(x)$, the polynomial scalar variables take the form $p_{i0}(x) = c_{i0}^Tm(x)$, and the polynomial vector variables take the form $p_{i}(x)=c_{i}^Tm(x)$ for $i=1,2$.
\end{theorem}

\begin{proof}
Applying Theorem 1, we require the existence of a solution to (\ref{Theorem1}). In Theorem 1 we had $s_1,s_2\in\Sigma[x]$. This creates an optimization problem over the cone of SOS polynomials. Here, to create an LP, we inner-approximate the SOS cone with the DSOS cone, as introduced in Section 2.3. By taking $s_1,s_2\in\mathcal{D}[x]$, each constraint of the form $Q_i\succeq_{dd}0$ is enforced with the linear diagonal dominance constraints from (\ref{DDconst}). Thus, there is no dependence on the indeterminate $x\in\mathbb{R}^n$ when constraining $s_1,s_2\in\mathcal{D}[x]$.

Likewise, the equation constraining $G_k$ in (\ref{Theorem3}) has no dependence on the indeterminate $x\in\mathbb{R}^n$. The Gram matrix operation, $G_k$, removes dependence on the indeterminate while preserving equality with the polynomial argument. Because $b$, $L_fb$, and $L_gb$ have no dependence on $z$, each entry of $G_k$ is trivially a linear expression of the decision variables. Therefore, each constraint in (\ref{Theorem3}) is linear in the decision variable $z$. As well, since the objective function is a constant $0$, the program is linear.

Thus, if there is a solution to the linear program (\ref{Theorem3}), then the conditions of Theorem 1 are satisfied and the system is invariant to safe set $\mathcal{B}$, verifying $b$ as a CBF.
\end{proof}

\section{Multiple CBF Linear Programming}
In this section we solve Problem 3. We consider the translation of the multiple candidate verification statements derived in Theorem 2 to linear programs. First, we establish a lemma that is utilized in the statement transcription process.

\begin{lemma} The cone of DSOS polynomials $\mathcal{D}[x]$ in $\mathbb{R}[x]$ is a generating preprime.
\end{lemma}

\begin{proof} For any cone $\mathcal{C}\subseteq\mathbb{R}[x]$, $\mathcal{C}$ is closed under addition, closed under multiplication, and $a^2\in\mathcal{C}$ if $a\in \mathbb{R}[x]$ \cite{Bochnak_RealAlgebraicGeometry}. The cone $\mathcal{D}[x]\subseteq\mathbb{R}[x]$ is closed under addition and multiplication, and $1\in\mathcal{D}[x]$ since $1\in\mathbb{R}[x]$. By definition, $\mathcal{D}[x]$ is a full-dimensional cone, i.e., $\mathcal{D}[x]-\mathcal{D}[x]=\mathbb{R}[x]$ \cite{Ahmadi_DSOSAndSDSOSOptimization}. Thus, $\mathcal{D}[x]$ is a generating preprime in $\mathbb{R}[x]$ \cite{Powers_CertificatesOfPositivity}.
\end{proof}

We define $M_{S,D}$ in an analogous manner to $M_S$ defined in (\ref{MSdef}), but with the SOS polynomials replaced by DSOS polynomials. Then a polynomial $q\in M_{S,D}$ if and only if $q$ is of the form
\begin{equation}\label{Dx-Module}
    q= s_0(x) + \sum_{i=1}^L s_i(x)b_i(x),
\end{equation}
where $s_i\in \mathcal{D}[x]$ for $i=0,\dots,L$. Thus, $M_{S,D}$ is a $\mathcal{D}[x]$-module.

In the following theorem statement, we assume that $M_{S,D}$ is Archimedean. The Archimedean property requires that the diameter of the set $\mathcal{B}$ is bounded by a natural number, as defined in Proposition 1. When the assumption is not already satisfied by the polynomials in $\mathcal{B}$, it is simple to enforce. By augmenting $S$ with the function $C-\sum_{i=1}^nx_i^2\geq 0$, with $C\in\mathbb{N}\backslash\{0\}$ being a bound on the diameter of $\mathcal{B}$, the quadratic module with respect to the augmented set $S'=S\cap \{x\in\mathbb{R}^n\mid C-\sum_{i=1}^nx_i^2\geq 0\}$ is Archimedean.

\begin{theorem}
Given a system (\ref{sysdef}) and a collection of candidate CBFs, $S=\{b_1,\dots,b_L\}$, assume that $M_{S,D}$ satisfies the Archimedean property. Let $G_k$ be as defined by (\ref{gram}) with $k=\lceil d/2 \rceil$, where $d$ is the degree of its polynomial argument. If the following conditions are satisfied, then for all $i\in[L]$ $b_i\in S$ is a CBF and $\mathcal{B}\not=\emptyset$:
\begin{enumerate}
    \item Each of the $L$ candidate CBFs must independently produce solutions to the LP (\ref{Theorem3}), satisfying  the conditions of Theorem 3.
    \item Let $s_i\in\mathcal{D}[x]$ for every $i\in[L]$. The following LP must not have a solution:
    \begin{equation}\label{Theorem4}
        \begin{aligned}
        \min_z \quad &0 \\
        \text{s.t.} \quad 
        &Q_i\succeq_{dd}0 \quad \quad \text{for all } i = 0,\dots,L \\
        &G_k(1 + s_0 + \sum_{i=1}^L s_ib_i) =0,
        \end{aligned}
    \end{equation}
    where the decision variable is
    \begin{equation}
    \begin{aligned}
        z = [ \text{uTri}(Q_0)^T \;\; \text{uTri}&(\tau_0)^T \;\; \dots \\ \;\; &\text{uTri}(Q_L)^T \;\; \text{uTri}(\tau_L)^T]^T,
    \end{aligned}
    \end{equation}
    with $z\in\mathbb{R}^{(k^2+k)(L+1)}$, and where the DSOS variables are defined as $s_i(x)=m^T(x)Q_im(x)$ for $i=[L]\cup\{0\}$. The variables $\tau_i$ for $i=[L]\cup\{0\}$ are the symmetric bounding matrices used to linearize the absolute value constraints from (\ref{DDconst}).
\end{enumerate}
\end{theorem}

\begin{proof} 
Applying Theorem 2, each CBF candidate must independently satisfy (\ref{Theorem1}) to be verified as a CBF. 
In Theorem 2 we had $s_i(x)\in\Sigma[x]$ for $i=0,\dots,L$. Here, we inner-approximate the SOS cone with the DSOS cone. As explained in the proof of Theorem 3, using $s_i\in\mathcal{D}[x]$ for $i=0,\dots,L$ transforms (\ref{MultiC2}) from an SDP to an LP due to the diagonal dominance constraints from (\ref{DDconst}) being linear inequalities. Similarly, the equations constraining $G_k$ in (\ref{Theorem4}) have no dependence on the indeterminate $x\in\mathbb{R}^n$. Since each $b_i$ has no dependence on $z$, the polynomial argument of $G_k$ is a linear expression of the decision variables. Thus, (\ref{Theorem4}) is a linear program.
 
From Theorem 2, we require that a solution does not exist to (\ref{MultiC2}) to ensure that $-1\not\in M_S$. As established in Lemma 1, $\mathcal{D}[x]$ is a generating preprime in $\mathbb{R}[x]$ and $M_{S,D}$ is a $\mathcal{D}[x]$-module. Then $M_{S,D}$ is an Archimedean $\mathcal{D}[x]$-module due to the Archimedean assumption imposed in the theorem statement. By Proposition 4, $-1\not\in M_{S,D}$ implies $\mathcal{B}\not=\emptyset$. Then the lack of solution to the LP in (\ref{Theorem4}) implies $\mathcal{B}\not=\emptyset$. 
\end{proof}

The safety verification program scales linearly with the number of candidate CBFs,
with $L$ candidate CBFs there are $L+1$ linear programs. In total, there are $3L+1$ DSOS polynomial decision variables across the $L+1$ LPs that need to be solved in Theorem 4. We note that the multiple CBF verification problem can be solved using the preorder instead of the quadratic module. However, doing this leads to an exponential growth rate in the number of DSOS polynomial decision variables. The Archimedean assumption allows for the use of the quadratic module and the improved computational transcription that results. 

\section{Simulations}

The proposed method for multiple CBF verification was simulated in MATLAB, and the SPOT toolbox developed in \cite{Ahmadi_DSOSAndSDSOSOptimization} was used for both SOS and DSOS
problems.\footnote{Code available at: https://github.com/elliempond/

SatelliteInspectionCBF.git} SPOT is a parser that transforms SOS problems to semidefinite programs and DSOS programs to linear programs, and we use it for both to provide a fair comparison
between the amounts of computation time they require. 
We used the SDP solver SeDuMi for both the SOS and DSOS programs \cite{Sedumi}. The simulations were computed on a workstation running Ubuntu 20.04 with a 3.30 GHz Intel Core processor and 128 GB of RAM.

We consider a satellite inspection problem with~$L$ chaser satellites indexed over $[L]$ and one target satellite in a circular orbit, depicted in
Figure~\ref{fig:satellite}. The objective of this problem is for each of the $L$ chaser satellites to approach the target satellite with a certain relative position. We use the linearized Clohessy-Wiltshire equations of motion of relative position and velocity for each of the $i\in[L]$ chaser satellites given by \cite{Curtis_OrbitalMech}
\begin{equation}
    \begin{aligned}
    \ddot{x}_i &= 2n\dot{y}_i+ 3n^2 + \frac{F_{x,i}}{m_i} \\
    \ddot{y}_i &= -2n\dot{x}_i + \frac{F_{y,i}}{m_i} \\
    \ddot{z}_i &= -n^2z_i + \frac{F_{z,i}}{m_i}, 
    \end{aligned}
\end{equation}
with state $[x_i \;\; y_i \;\; z_i \;\; \dot{x}_i \;\; \dot{y}_i \;\; \dot{z}_i]^T$, input $[F_{x,i}\;\;F_{y,i}\;\;F_{z,i}]^T$, mass of the chaser satellite $m_i$, and mean motion $n$. The mean motion is the average angular speed of a satellite during one full orbit. We let $r_i=[x_i \;\; y_i \;\; z_i]^T$ be the relative position vector between chaser $i$ and the target consisting of the radial, tangential, and out-of-plane relative distances, respectively. This admits dynamics in the form of (\ref{sysdef}). For each $i\in[L]$, we form barrier functions between chaser $i$ and the target satellite as 
\begin{equation}\label{simb}
    b_i(z) = r_i^Tr_i  +\frac{m_i}{T_{i}}\dot{r}_i^T\dot{r}_i - R_t^2,
\end{equation}
where $R_t$ is the minimum safe radius about the target satellite and $T_i$ is the nominal thrust of each chaser. 

\begin{figure}
\begin{center}
\includegraphics[width=8.4cm]{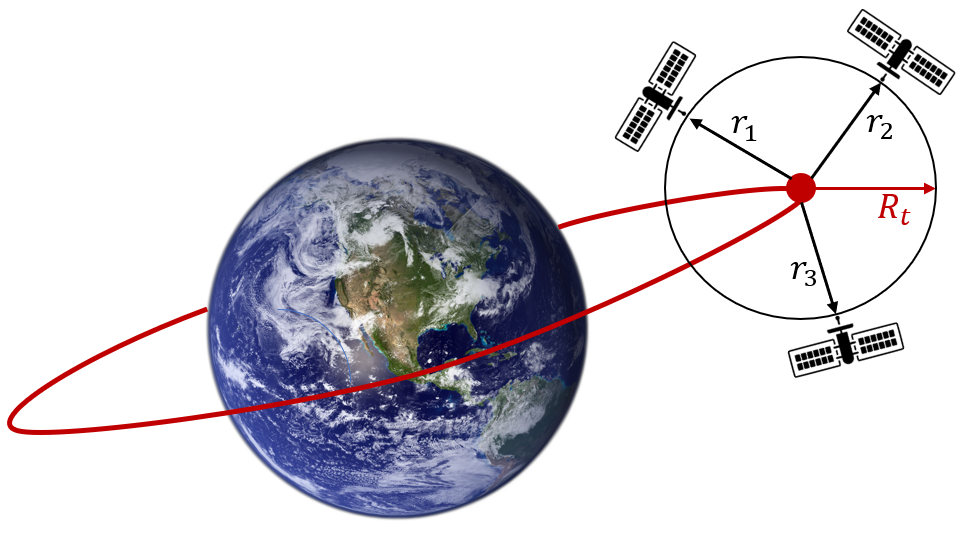}    
\caption{The satellite inspection problem depicted with $L=3$ chaser satellites. The target satellite is represented by the red dot in a circular orbit.} 
\label{fig:satellite}
\end{center}
\end{figure}

Following Theorem 4, this CBF verification 
problem consists of $L$ linear programs in the form of (\ref{Theorem3}). These LPs verify that each candidate of the form of (\ref{simb}) is a CBF. Additionally, we have one linear program in the form of (\ref{Theorem4}) that verifies that $\mathcal{B}=\{z\in\mathbb{R}^{6L}\mid b_1(z)\geq 0,\dots, b_L(z)\geq0\} \not=\emptyset$. We consider 3U cubesat specifications for the chaser satellites with a mass of $m_{i}=2$ kg, nominal thrust $T_i= 0.5$ N for all $i\in[L]$, and the mean motion for a circular low earth orbit $n=0.0010$ rad/s. We consider close approach and proximity operations at $R_t=0.5$ km.

\begin{table}[h]
\begin{center}
\caption{Computation time.}\label{tb:table1}
\begin{tabular}{cccc}
$L$ & DSOS (s) & SOS (s) & \% Reduction \\\hline
$1$ & $1.0$ & $1.7$ & $41.2\%$ \\
$2$ & $2.1$ & $8.4$ & $75.0\%$ \\
$3$ & $4.3$ & $29.5$ & $85.4\%$ \\
$4$ & $10.2$ & $114.8$ & $91.1\%$ \\
$5$ & $25.0$ & $499.0$ & $95.0\%$ \\
$6$ & $63.0$ & $1215.4$ & $94.8\%$ \\  
$7$ & $156.5$ & $3423.7$ &  $95.4\%$\\
\hline 
\end{tabular}
\end{center}
\end{table}

Table 1 gives the computation time for up to $7$ chaser satellites. For $L=1$ satellite, only Theorem~3 was required for verification. There is little improvement between the DSOS linear program and the SOS semidefinite program for one satellite. However, the improvement is immediate for $L\geq2$. The DSOS programs required significantly less time to obtain the same affirmative verification as the corresponding SOS program. While there can be a loss of accuracy when inner-approximating 
SOS programs 
with DSOS LPs, it was not seen with this example problem in the sense that the DSOS program produced an affirmative verification for all scenarios that were considered,
and this verification is in agreement with the SOS program verification. 

These results show that the transcription of CBF verification to linear programs
provides a substantial reduction in computation time while providing
the same feasibility assurances as the more complex SOS verification formulation. 
In fact, Table~1 shows that for $L \geq 4$ chaser satellites the LP formulation
reduces computation time by more than~$90$\%, which is a significant reduction. 
These results
suggest that the LP formulation can not only reduce the computations
required for CBF verification, but also provide new capabilities to verify
larger problems.

\section{Conclusion}
We have presented, to the best of the authors' knowledge, the first verification procedure for
the validation of candidate control barrier functions using linear programming. We have also contributed a method of multiple CBF verification that scales linearly with the number of candidate functions. In future work, these verification methods will be extended to high-order CBFs, dynamic systems with input constraints, and distributed linear CBF certification in multi-agent systems.


\bibliography{ifacconf}             
                                                   







\end{document}